\documentclass[a4paper, reqno]{amsart}
\usepackage[colorlinks=true,linkcolor=blue]{hyperref}    
\usepackage{amssymb, amsfonts, amsthm, amsmath, mathtext, latexsym,graphicx}
\newcommand{\D}{\displaystyle}
\newcommand{\R}{{\mathbb R}}
\newcommand{\Z}{{\mathbb Z}}
\newcommand{\N}{{\mathbb N}}

\newcommand{\C}{{\mathbb C}}
\newcommand{\re}{{\rm Re}\,}
\newcommand{\im}{{\rm Im}\,}
\theoremstyle{plain}
\newtheorem{theorem}{Theorem}[section]
\newtheorem{lemma}{Lemma}[section]
\newtheorem{proposition}{Proposition}[section]

\theoremstyle{definition}
\newtheorem{remark}{Remark}[section]
\numberwithin{equation}{section}
\usepackage[noabbrev]{cleveref}
\title{Semiclassical asymptotics of meromorphic solutions 
of difference equations}
\author{Alexander Fedotov and Fr{\'e}d{\'e}ric Klopp}
\address[Alexander Fedotov]{St. Petersburg State University, 
7/9 Universitetskaya nab., St.Petersburg, 199034, Russia}
\email{
  {a.fedotov@spbu.ru}}
\address[Fr{\'e}d{\'e}ric Klopp]{Sorbonne Universit{\'e}, Universit{\'e} Paris Diderot,
  CNRS, Institut de Math{\'e}matiques de Jussieu - Paris Rive Gauche ,
  F-75005, Paris, France}
\email{
  {frederic.klopp@imj-prg.fr}}
\thanks{The work was supported by PRC Russie CNRS, France, and the
  Russian Foundation for Basic Research under the grant No
  17-51-150008.}
\begin{document}
\begin{abstract}
  We consider the difference Schr{\"o}dinger equation
  $\psi(z+h)+\psi(z-h)+ v(z)\psi(z)=E\psi(z)$ where $z$ is a complex 
  variable, $E$ is a spectral parameter, and $h$ is a small positive 
  parameter. If the potential $v$ is an analytic
  function, then, for $h$ sufficiently small, the analytic solutions
  to this equation have standard semi-classical behavior that can be
  described by means of an analog of the complex WKB method for
  differential equations.  In the present paper, we assume that $v$
  has a simple pole and, in its neighborhood, we study the asymptotics
  of meromorphic solutions to the difference Schr{\"o}dinger equation.
\end{abstract}
\maketitle
\section{Introduction}
We study the difference Schr{\"o}dinger equation
\begin{equation}\label{main}
  \psi(z+h)+\psi(z-h)+v(z)\psi(z)=E\psi,
\end{equation}
where $z$ is a complex variable, $v$ is a given meromorphic function
called {\it potential}, $E$ is a {\it spectral} parameter, and $h$ is a small 
positive {\it shift} parameter \footnote{The results were announced 
in the conference proceedings~\cite{F-K:18a}}.
\\
Instead of~\eqref{main}, one often considers equations of the form
\begin{equation}\label{main:1}
  \phi_{k+1}+\phi_{k-1}+v(kh+\theta)\phi_k=E\phi_k, 
\end{equation}
where $k\in\Z$ is an integer variable, and $\theta\in\R $ is a
parameter. There is a simple relation between~\eqref{main}
and~\eqref{main:1}: if $\psi$ is a solution to~\eqref{main}, then the
formula $\phi_k=\psi(kh+\theta)$ yields a solution
to~\eqref{main:1}. We note that, when $h$ is small, the coefficient in
front of $\phi_k$ in~\eqref{main:1} varies slowly in $k$.
\\
Formally,
$\psi(z+h)=\sum_{l=0}^\infty \frac{h^l}{l!}
\frac{d^l\psi}{dz^l}(z)=e^{h\frac{d}{dz}}\psi(z)$; thus,
in~\eqref{main} $h$ is a small parameter in front of the
derivative. So, $h$ is a standard semiclassical parameter.
\\
The semi-classical asymptotics of solutions to ordinary differential
equations, e.g., the Schr{\"o}dinger equation
\begin{equation}\label{differential-eq}
  -h^2\frac{d^2\psi}{dz^2}(z)+v(z)\psi(z)=E\psi(z),
\end{equation}
are described by means of the well-known WKB method (called so after
G. Wentzel, H. Kramers and L. Brillouin).  There is a huge literature
devoted to this method and its applications.  If $v$
in~\eqref{differential-eq} is analytic, one uses the variant often
called the complex WKB method (see, e.g.,~\cite{Wa:87, Fe:93}). This
powerful and classical asymptotic method is used to study solutions
to~\eqref{differential-eq} on the complex plane. Even when studying
this equation on the real line, the complex WKB method is used to
compute exponentially small quantities (such as the overbarier
tunneling coefficient or the exponentially small lengths of spectral
gaps in the case of a periodic $v$) or to simplify the asymptotic
analysis (e.g. by going to the complex plane to avoid turning points),
see, e.g.,~\cite{Fe:93}.  The case of meromorphic coefficients is a
classical topic in the complex WKB theory. The analog of the complex
WKB method for difference equations is being developed
in~\cite{B-F:94, F-Shch:18, F-K:18} and in the present paper, where we
turn to meromorphic solutions to~\cref{main}.
\\
Difference equations~\eqref{main} on $\R$ or on $\C$
and~\eqref{main:1} on $\Z$ with a small $h$ arise in many fields of
mathematics and physics.  In quantum physics, for example, one
encounters such equations when studying in various asymptotic
situations an electron in a two-dimensional crystal submitted to a
constant magnetic field (see, e.g.~\cite{Wi} and references
therein). The electron is described by a magnetic Schr{\"o}dinger operator
with a periodic electric potential.  And, for example, in the
semi-classical limit, in certain cases its analysis asymptotically
reduces to analyzing an $h$-pseudo-differential operator with the
symbol $H(x,p)= 2\cos p+2\cos x$ (see~\cite{H-S:88}). Its
eigenfunctions satisfy \cref{main} with $v(z)=2\cos z$. The parameter
$h$ is proportional to the magnetic flux through the periodicity cell,
and the case when $h$ is small is a natural one.  The reader can find
more references and examples in~\cite{GHT:89}.  We add only that
\cref{main:1} for $v(z)=\lambda \cos(2\pi z)$, $\lambda$ being a
coupling constant, is the famous almost Mathieu equation (see,
e.g.,~\cite{A-J:09}), and, for $v(z)=\lambda \cot(\pi z)$, it is the
well known Maryland equation (Maryland model) introduced by
specialists in solid state physics in~\cite{GFP}.
\\
Difference equations in the complex plane (with analytic or
meromorphic coefficients) arise in many other fields of mathematics
and physics, in particular, in the study of the diffraction of
acoustic waves by wedges (see, e.g.,~\cite{B-L-G:2008}) or in the
theory of differential quasi-periodic equations (see,
e.g.,~\cite{F-K:02}). Small shift parameters arise in the problem of
diffraction by thin wedges (the shift parameter is proportional to the
angle of the wedge (see~\cite{B-L-G:2008})) and for quasi-periodic
equations with two periods of small ratio (the shift parameter is
proportional to this ratio (see, e.g.,~\cite{F-K:02})). The
semi-classical analysis of difference equations is also used to study
the asymptotics of orthogonal polynomials (see,
e.g.,~\cite{G-at-al,Dobro,Wo:03}).
\\
Even when studying~\eqref{main:1} on $\Z$, it is quite natural to pass
to the analysis of~\eqref{main} on $\R$ or on $\C$ as for this
equation one can fruitfully use numerous analytic ideas developed in
the theory of differential equations, e.g., tools of the theory of
pseudo-differential operators and of the complex WKB method. If the
coefficient $v$ is periodic, for equation~\eqref{main} one can use
ideas of the Floquet theory for differential equations with periodic
coefficients, which leads to a natural renormalization method
(see~\cite{F:13}).
\\
B. Helffer and J. Sj{\"o}strand (e.g., in~\cite{H-S:88}) and V. Buslaev
and A. Fedotov (see, e.g.~\cite{F:13}) studied the cantorian
geometrical structure of the spectrum of the Harper operator in the
semiclassical approximation.  Therefore, V. Buslaev and A. Fedotov
began to develop the complex WKB method for difference equations
in~\cite{B-F:94}. We are going to use the results of the present paper
to study in the semiclassical approximation the multiscale structure
of the (generalized) eigenfunctions of the Maryland operator (by means
of the renormalization method described in~\cite{F-S:15}). In the
``anti-semiclassical'' case, for the almost Mathieu operator such a
problem was solved in~\cite{J-W}.
\\
In this paper, for small $h$, we describe uniform asymptotics of
meromorphic solutions to~\cref{main} near a simple pole of $v$.  In
the case of a differential equation, say,~\cref{differential-eq} with
a meromorphic $v$, the solutions may have singularities (branch points
or isolated singular points) only at poles of $v$. In the case
of~\cref{main}, the behavior of its solutions is completely different.
\\
Let $d_x>0$, $d_y>0$, and
$S=\{z\in\C\,:\,|\re z|<d_x, \,|\im z|<d_y\}$.  We assume that $v$ is
analytic in $S\setminus\{0\}$ and has a simple pole at zero. Let
$\psi$ be a solution to~\cref{main} that is analytic in
$\{z\in S\,:\,\re z<0\}$. \Cref{main} implies that
$\psi(z)=-\psi(z-2h)-v(z-h)\psi(z-h)$. Therefore, for sufficiently
small $h$, \ $\psi$ can be meromorphically continued into $S$. It may
have poles at the points of $h\N$.  When $h$ becomes small, these
points become close one to another. We describe the semi-classical
asymptotics of such meromorphic solutions in $S$.\\
Below, unless stated otherwise, the estimates of the error terms in
the asymptotic formulas are locally uniform for $z$ in the domain
that we consider (i.e., uniform on any given compact subset of such 
a domain).\\
Instead of saying that an asymptotic representation is valid 
for sufficiently small $h$, we write that it is valid as $h\to0$.\\
In the sequel, we shall not distinguish between a meromorphic
function and its meromorphic continuation to a larger domain.\\
We also use the notations
$\R_\pm=\{z\in\C\,:\, \im z=0, \ \pm\re z\ge 0\}$,
$\R_\pm^*=\R_\pm\setminus{0}$ and $\C_\pm=\{z\in\C\,:\, \pm\im z>0\}$.
\section{Main results}
\subsection{The complex WKB method in a nutshell}
\label{intro:WKB}
Formally~\cref{main} can be written in the form
\begin{gather}\label{differential}
  (2\cos\hat p+v(z))\psi(z)=0,\quad \hat p=-ih\frac{d}{dz}\;.
\end{gather}
One of the main objects of the complex WKB method is {\it the complex
  momentum} $p$ defined by the formula
\begin{equation}\label{eq:p}
  2\cos\,p(z)+v(z)=0.
\end{equation}
It is an analytic multivalued function. Its branch points are
solutions to $v(z)=\pm 2$. The points where $v(z)=\pm 2$ 
are called {\it turning points}. A subset $D$ of the domain 
of analyticity of $v$ is {\it regular} if it contains no turning 
points.\\
Let $D$ be a regular simply connected domain, and $p_0$ be a branch of
the complex momentum analytic in $D$. Any other branch of the complex
momentum that is analytic in $D$ is of the form $s p_0+2\pi m$ for
some $s\in\{\pm 1\}$ and $m\in\Z$.\\
In terms of the complex momentum, one defines {\it canonical} domains.
The precise definition of a canonical domain can be found
in~\cref{ss:basics-of-WKB}. Here, we note only that the canonical domains
are regular and simply connected and that any
regular point is contained in a canonical domain (independent of $h$).\\
One of the basic results of the complex WKB method is
\begin{theorem}[\cite{F-Sh:17}]
  \label{th:wkb_main}
  Let $K\subset \mathbb C$ be a bounded canonical domain; let $p$ be a branch
  of the complex momentum analytic in $K$ and pick $z_0\in K$. For
  sufficiently small $h$, there exists $\psi$, a solution
  to~\cref{main} analytic in $K$ and such that, in $K$, one has
  \begin{equation}
    \label{standard_asymptotic}
    \psi(z)=\frac1{\sqrt{\sin p(z)}}\, e^{\frac{i}{h}\int_{z_0}^z
      p(\zeta)\,d\zeta+o(1)},\qquad h\to 0, 
  \end{equation}
  where $z\mapsto\sqrt{\sin(p(z))}$ is analytic in $K$.
\end{theorem}
\begin{remark} \label{rem:sin} The function $\sin p$ does not vanish
  in regular domains; indeed, $\sin p$ only vanishes at the points
  where $v(z)=-2\cos p(z)\in\{\pm 2\}$, i.e., at the turning points.
\end{remark} 
\noindent In the case of the Harper equation (for unbounded canonical
domains), the analog of Theorem~\ref{th:wkb_main} was proved
in~\cite{B-F:94}.\\
Let us underline that the branch $p$ of the complex momentum in
Theorem~\ref{th:wkb_main} need not be the one with respect to which
$K$ is canonical. 
\subsection{Asymptotics of a meromorphic solution}
\label{ss:main}
Let us turn to the problem discussed in the present paper. Recall that
0 is a simple pole of $v$. Since $v(z)\to\infty$ as $z\to0$, reducing
somewhat $d_x$ and $d_y$ if necessary, we can and do assume that the
set $S\setminus \{0\}$ is regular and that the imaginary part of the
complex momentum does not vanish there.
\subsubsection{The solution we study}
\label{sss:main}
Let $S'=S\setminus{\mathbb R}_+$. In $S'$, fix an analytic branch $p$
of the complex momentum satisfying $\im p(z)<0$.\\
Pick a point $z_0$ in $S'\cap \R_-$. As this point is regular, there
exists a solution $\psi$ to~\cref{main} that is analytic in a
neighborhood of $z_0$ independent of $h$ and that admits the
asymptotic representation~\eqref{standard_asymptotic} in this
neighborhood.\\
Adjusting $d_x$ and $d_y$ if necessary, we can and do assume that
there exists $c\in (0, d_x)$ such that $\psi$ is analytic and admits the
asymptotic representation~\eqref{standard_asymptotic} in the domain
$S_c=\{z\in S\,:\, \re z<-c\}$.\\
As $\im p(z)<0$ in $S'$, the expression
$z\mapsto \left| e^{\frac{i}{h}\int_{z_0}^z p(\zeta)\,d\zeta}\right|$
(compare with the leading term in~\eqref{standard_asymptotic})
increases as $z$ in $S'$ moves to the right parallel to $\R$.\\
If $h$ is sufficiently small, the solution $\psi$ is meromorphic in
$S$; its poles belong to $h\N$ and they are simple.
\subsubsection{The uniform asymptotics of $\psi$ in $S$}
\label{sec:asympt-psi-whole}
To describe the asymptotics of $\psi$, we define an auxiliary
function. Clearly, the complex momentum $p$ has a logarithmic branch
point at zero. In $\C\setminus\R_+$, we fix the analytic branch of
$z\mapsto \ln(-z)$ such that $\ln(-z)|_{z=-1}=0$.
In~\cref{sec:case-where-im-1} we check
\begin{lemma}
  \label{le:p-ln}
  The function $z\mapsto p(z)-i\ln (-z)$ is analytic in $S$.  The
  function $z\mapsto z\sin p(z)$ is analytic and does not vanish in
  $S$.
\end{lemma}
\noindent For $z\in S'$, we set
\begin{equation}
  \label{eq:G0}
  G_0(z)=\frac{\sqrt{h/2\pi}}{\sqrt{- z\sin p(z)}} \,  
  e^{\textstyle\ \frac{z}{h}\ln\frac1h+
    \frac{i}h\int_0^z(p(\zeta)-i\ln(-\zeta))\,d\zeta}.
\end{equation}
Here and below, $\sqrt{h/2\pi}$ and $\ln\frac1h$ are positive;
$\sqrt{- z\sin p(z)}= \sqrt{-z}\sqrt{\sin p(z)}$; the branch of 
$z\mapsto\sqrt{\sin p(z)}$ coincides with the one
from~\eqref{standard_asymptotic}.\\
In view of \Cref{le:p-ln}, $G_0$ is analytic in $S$.\\
Our main result is
\begin{theorem}
  \label{th:main}
  In $S$, the solution $\psi$ admits the asymptotic representation
  \begin{equation}\label{eq:psi-gamma}
    \psi(z)=\Gamma\,\left(1-\frac{z}h\right)\,G_0(z)\,
    e^{\textstyle\frac{i}h\int_{z_0}^0p\,dz+o(1)},\quad\text{ as
    }h\to 0
  \end{equation}
  where $\Gamma$ is the Euler $\Gamma$-function and the integration
  path stays in $S$.
\end{theorem}
\noindent So, the special function describing the asymptotic behavior of
$\psi$ near the poles generated by a simple pole of $v$ is
the $\Gamma$-function.
\subsubsection{The asymptotics of $\psi$ outside a neighborhood of
  $\R_+$}
\label{sec:asympt-psi-outs}
For large values of $|z/h|$, the $\Gamma$-function
in~\eqref{eq:psi-gamma} can be replaced with its asymptotics.  Let us
give more details.\\
Fix $\epsilon>0$. We recall that, uniformly in the sector
$|\arg \zeta|\le \pi -\epsilon$, one has
\begin{equation}
  \label{eq:stirling}
  \Gamma(1+\zeta)=\sqrt{2\pi\zeta}\,e^{\zeta (\ln\zeta -1)+o(1)}, \
  |\zeta|\to\infty,
\end{equation}
where the functions $\zeta\mapsto \sqrt{\zeta}$ and
$\zeta\mapsto\ln\zeta$ are analytic in this sector and satisfy the
conditions $\sqrt{1}=1$ and $\ln 1=0$.\\
Fix $\delta$ positive sufficiently small. Using~\eqref{eq:stirling},
one checks that, in $S$ outside the $\delta$-neighborhood of $\R_+$,
the representation~\eqref{eq:psi-gamma} turns
into~\eqref{standard_asymptotic}.\\
By construction, $\psi$ admits the asymptotic
representation~\eqref{standard_asymptotic} in $S_c$. \Cref{th:main}
implies that the representation remains valid in $S'$. This reflects
the standard semi-classical heuristics saying that an asymptotic
representation of a solution remains valid as long as the leading term
is increasing; in the present case, the modulus of the exponential in
the leading term in~\eqref{standard_asymptotic} increases in $S'$ as
long as $z$ moves to the right parallel to $\R$.
\subsubsection{The asymptotics of $\psi$ near $\R_+$ away from $0$}
\label{sec:asymptotics-psi-near}
Assume that $z$ is inside the $\delta$-neighborhood of $\R_+$ but
outside the $\delta$-neighborhood of 0.  In this case, to
simplify~\eqref{eq:psi-gamma}, we first use the relation
\begin{equation}
  \label{eq:two-gammas}
  \Gamma(1-\zeta)=\frac{\pi}{\sin(\pi\zeta)}\,\frac1{\Gamma(\zeta)}
\end{equation}
and, next, the asymptotic representation~\eqref{eq:stirling}. This
yields the asymptotic representation
\begin{equation}
  \label{as:in-out}
  \psi(z)=\frac{e^{\textstyle \
      \frac{i}{h}\int_{z_0}^zp(\zeta)\,d\zeta+o(1)}}{(1-e^{2\pi iz/h})\;
    \sqrt{\sin(p(z))}}\,,\quad h\to0,
\end{equation}
where $p$, $z\mapsto\int_{0}^zp(\zeta)\,d\zeta$ and $\sqrt{\sin(p)}$
are obtained by analytic continuation from $S'\cap \C_+$ into the
domain under consideration.
\subsection{A basis of solutions}
\label{intro:second-sol}
The set of solutions to~\eqref{main} is a two-dimensional module over
the ring of $h$-periodic functions (see
section~\ref{s:space-of-sol}). We now explain how to construct a basis
of this module.

\subsubsection{} As the first solution, we take
$f_+(z)=e^{-\frac{i}{h}\int_{z_0}^0p(z)\,dz}\,\psi(z)$. In $S'$, it
admits the asymptotic representation
\begin{equation}\label{eq:f-plus}
  f_+(z)\sim \frac1{\sqrt{\sin p(z)}}\, 
  e^{\textstyle \;\frac{i}{h}\int_{0}^z p(\zeta)\,d\zeta},  \qquad h\to 0,
\end{equation}
and has simple poles at the points of $h\N$.
\\
We note that $\left|e^{\frac{i}{h}\int_{0}^z p(\zeta)\,d\zeta}\right|$
increases when $z$ moves to the right parallel to $\R$.
\subsubsection{} Fix $z_1\in S\cap\R_+^*$. Possibly reducing $S$
somewhat, similarly to $\psi$, one constructs a solution $\phi$ that,
in $S\setminus \R_-$, admits the asymptotic representation
\begin{equation}
  \label{eq:phi}
  \phi(z)\sim\frac1{\sqrt{\sin(p(z))}}\, e^{-\frac{i}{h}\int_{z_1}^z
    p(\zeta)\,d\zeta},\qquad h\to 0,
\end{equation}
and has simple poles at the points of $-h\N$. The branches of the
complex momenta appearing in~\eqref{eq:f-plus} and~\eqref{eq:phi}
coincide in $\C_+$.
\\
Note that $\left|e^{-\frac{i}{h}\int_{z_1}^z p(\zeta)\,d\zeta}\right|$
increases when $z$ moves to the left parallel to $\R$.
\\
The function $z\mapsto 1-e^{2\pi iz/h}$ being $h$-periodic, we define
another solution to~\cref{main} by the formula
$f_-(z):=(1-e^{2\pi iz/h}) e^{\frac{i}{h}\int_{z_1}^0
  p(z)\,dz}\,\phi(z)$. The function $f_-$ is analytic in $S$; its
zeroes are simple and located at the points of $h\N\cup\{0\}$.  As we
prove in \cref{sec:second-solution}, in $S'$, the solution $f_-$ has
the asymptotics
\begin{equation}
  \label{eq:f-minus}
  f_-(z)\sim\frac1{\sqrt{\sin(p(z))}}\, e^{-\frac{i}{h}\int_{0}^z
    p(\zeta)\,d\zeta+o(1)},\qquad h\to 0.
\end{equation}
\subsubsection{} In section~\ref{sec:second-solution}, we shall see
that, for sufficiently small $h$, $f_+$ and $f_-$ form a basis of the
space of solutions to~\cref{main} meromorphic in $S$ (possibly reduced
somewhat).
\subsection{The idea of the proof of \Cref{th:main} and the plan of
  the paper}\label{ss:plan}
To prove Theorem~\ref{th:main}, we consider the function
$z\mapsto f(z)=\psi(z)/\Gamma(1-z/h)$.  It is analytic in $S$. Using
tools of the complex WKB method for difference equations, outside a
disk $D$ centered at 0 (and independent of $h$), we compute the
asymptotics of $f$ and obtain $f(z)=e^{\frac ih\int_{z_0}^0p\,dz}G_0(z)(1+o(1))$. 
The factor $G_0$ is analytic and does not vanish in $S$.  Therefore, the function
$z\mapsto e^{-\frac ih\int_{z_0}^0p\,dz}f(z)/G_0(z)-1$ is analytic in $D$, and, as it is small
outside $D$,  the maximum principle implies that it is small also inside $D$.\\
The plan of the paper is the following. In~\cref{s:preliminaries} we
describe basic facts on~\cref{main} and the main tools of the complex
WKB method for difference equations.  In~\cref{s:outside-as}, we
derive the asymptotics of the solution $\psi$ in $S$ outside a
neighborhood of 0. In~\cref{s:uniform-as}, we finally prove the
asymptotic representation~\eqref{eq:psi-gamma}. In~\cref{s:basis}, we
briefly discuss the solution $\phi$ mentioned
in~\cref{intro:second-sol}.
\subsection{Acknowledgments}
\label{sec:acknowledgments}
This work was supported by the CNRS and the Russian foundation of
basic research under the French-Russian grant 17-51-150008.
\section{Preliminaries}\label{s:preliminaries}
We first recall basic facts on the space of solutions to~\cref{main};
next, we recall basic constructions of the complex WKB method for
difference equations and prepare an important tool,
\Cref{th:rectangle}. We will use it various times to obtain the
asymptotics of solutions to~\eqref{main}.
\subsection{The space of solutions to~\cref{main}}
\label{s:space-of-sol}
The observations that we now discuss are well-known in the theory of
difference equations and are easily proved. We follow~\cite{F-K:18}.\\
Fix $(X_1,X_2, Y)\in {\mathbb R}^3$ so that $X_1+2h<X_2$.  We discuss the
set $\mathcal{M}$ of solutions to~\cref{main} on $I:=\{z\in \C\,:\,X_1<\re
z<X_2,\; \im z=Y\}$. \\
Let $\psi_\pm\subset \mathcal{M}$. The expression
\begin{equation}
  \label{eq:wronskian}
  w(\psi_+(z),\psi_-(z))=\psi_+(z+h)\psi_-(z)-\psi_+(z)\psi_-(z+h),
  \quad z,\,z+h\in I, 
\end{equation}
is called {\it the Wronskian} of $\psi_+$ and $\psi_-$.   It is
$h$-periodic in $z$.\\ 
If the Wronskian of $\psi_+$ and $\psi_-$ does not vanish, they form a
basis in $\mathcal{M}$, i.e, $\psi\in \mathcal{M}$ if and only if
\begin{equation}\label{eq:three-solutions}
  \psi(z)=a(z)\psi_+(z)+b(z)\psi_-(z),\quad z\in I,
\end{equation}
where $a$ and $b$ are $h$-periodic complex valued functions. One has
\begin{equation}
  \label{eq:periodic-coef}
  a(z)=\frac{w(\psi(z),\,\psi_-(z))}{w(\psi_+(z),\,\psi_-(z))}
  \quad\text{and}\quad 
  b(z)=\frac{w(\psi_+(z),\,\psi\,(z))}{w(\psi_+(z),\,\psi_-(z))}.
\end{equation} 
The set $\mathcal{M}$ is a two-dimensional module over the ring of
$h$-periodic functions.
\subsection{Basic constructions of the complex WKB method}
\label{ss:basics-of-WKB}
We begin by defining {\it canonical curves} and {\it canonical
  domains}, the main geometric objects of the method.
\subsubsection{Canonical curves}
\label{sec:canonical-curves}
For $z\in\C$, we put $x=\re z$, $y=\im z$. A connected curve
$\gamma \subset \mathbb{C}$ is called {\it vertical} if it is the graph of a
piecewise continuously differentiable function of $y$.\\
Define the complex momentum, turning points and regular domains as
in~\cref{intro:WKB}.\\
Let $\gamma$ be a regular vertical curve parameterized by $z=z(y)$, and
$p$ be a branch of the complex momentum that is analytic near $\gamma$.
We pick $z_0\in\gamma$. The curve $\gamma$ is called {\it canonical}
with respect to $p$ if, at the points where $z'(\cdot)$ exists, one
has
\begin{equation}\label{def:can}
  \frac{d}{dy} \im \int_{z_0}^zp(\zeta)\,d\zeta> 0
  \quad\text{ and }\quad 
  \frac{d}{dy} \im \int_{z_0}^z(p(\zeta)-\pi) \,d\zeta< 0,
\end{equation}
and at the points of jumps of $z'(\cdot)$, these inequalities are
satisfied for both the left and right derivatives.
\subsubsection{Canonical domains}
\label{sec:canonical-domains}
In this paper we discuss only bounded canonical domains.\\
Let $K\subset\C$ be a bounded simply connected regular domain and let
$p$ be a branch of the complex momentum analytic in it.  The domain
$K$ is said to be {\it canonical} with respect to $p$ if, on the
boundary of $K$, there are two regular points, say, $z_1$ and $z_2$
such that, for any $z\in K$, there exists a curve $\gamma\subset K$ passing
through $z$ and connecting $z_1$ to $z_2$ that is canonical with
respect to $p$.
In this paper, we use the local canonical domains described by
\begin{lemma}
  \label{le:1}
  For any regular point, there exists a canonical domain that contains
  this point.
\end{lemma}
\noindent This lemma is an analog of Lemma 5.3 in~\cite{F-K:05};
mutatis mutandis, their proofs are identical.
\subsubsection{Standard asymptotic behavior}
\label{sec:stand-asympt-behav}
Let $U\subset \C$ be a regular simply connected domain; pick $z_0\in U$ and assume
$z\mapsto p(z)$ and $z\mapsto \sqrt{\sin p(z)}$ are analytic in $U$.\\
We say that a solution $\psi$ to~\cref{main} has the standard
(asymptotic) behavior
\begin{equation}
  \label{st_beh}
  \psi(z)\sim \frac1{\sqrt{\sin(p(z))}}\, e^{\frac{i}{h}\int_{z_0}^z
    p(\zeta)\,d\zeta} 
\end{equation}
in $U$ if, for $h$ sufficiently small, $\psi$ is analytic
and admits the asymptotic representation~\eqref{standard_asymptotic}
in $U$.\\
\Cref{th:wkb_main} says that, for any given bounded canonical domain
$K$, for any branch of $z\mapsto p(z)$ analytic in $K$, there exists a
solution with the standard asymptotic behavior~\eqref{st_beh} in
$K$. To study its asymptotic behavior outside $K$, we use the
construction described in the next subsection.
\subsection{A continuation principle}
\label{sec:cont-princ}
Assume the potential $v$ in~\cref{main} is analytic in a domain in
$\C$.  Let $z_0$ be a regular point, $V_0$ be a regular simply
connected domain containing $z_0$ and $p$ be a branch of the complex
momentum analytic in $V_0$.\\
Finally, let $\psi$ be a solution to~\eqref{main} having the standard
asymptotic behavior~\eqref{st_beh} in $V_0$. One has
\begin{theorem}
  \label{th:rectangle}
  Let $z_1\in V_0$.  Consider the straight line
  $L=\{z\in{\mathbb C}\,:\,{\rm Im}\,z= {\rm Im}\, z_1\}$. Pick
  $z_2\in L$ such that ${\rm Re}\,z_2> {\rm Re}\,z_1$. Assume the segment
  $I=\{z\in L\,:\,{\rm Re}\,z_1\le {\rm Re}\,z\le {\rm Re}\,z_2\}$ is
  regular.\\
  If ${\rm Im}\, p\,(z)<0$ along $I$, then there exists $\delta>0$
  such that the $\delta$-neighborhood of $I$ is regular and $\psi$ has
  the standard behavior~\eqref{st_beh} in this neighborhood.
\end{theorem}
\noindent Theorem~\ref{th:rectangle} roughly says that the asymptotic
formula~\eqref{standard_asymptotic} stays valid along a horizontal
line as long as the leading term grows exponentially. It is akin to
Lemma 5.1 in~\cite{F-K:05} that deals with differential equations. The
proof of Theorem~\ref{th:rectangle} given below follows the plan of
the proof of Lemma 5.1 in~\cite{F-K:05}.\\
Let $\tilde \psi$ be a solution to~\cref{main} with the standard
behavior $\tilde \psi(z)\sim \frac{e^{-\frac{i}{h}\int_{z_0}^z
    p(\zeta)\,d\zeta}}{\sqrt{\sin(p(z))}}$ in $V_0$.  If $\im p<0$ in $V_0$,
then the analogue, mutatis mutandis, of Theorem~\ref{th:rectangle} on
the behavior of $\tilde \psi$ to the left of $V_0$ holds.
\begin{proof}[Proof of Theorem~\ref{th:rectangle}]
  Clearly, for $\zeta\in I$, there exists an open disk $D$ centered at
  $\zeta$ such that $D$ is regular and $\im p(z)<0$ in $D$. In view of
  \Cref{le:1}, if $D$ is sufficiently small, then there exists two
  solutions $\psi_\pm$ having the standard behavior
  $\psi_\pm(z)\sim \frac{e^{\pm\frac{i}{h}\int_{z_0}^z
      p(\zeta)\,d\zeta}}{\sqrt{\sin(p(z))}}$ in $D$ (here, we first
  integrate from $z_0$ to $z_1$ in $V_0$, next from $z_1$ to $\zeta$
  along $I$ and finally from $\zeta$ to $z$ inside $D$).\\
  The segment $I$ being compact, we construct finitely many open
  disks $(D_j)_{0\leq j\leq J}$, each centered in a point of $I$,
  covering $I$ and such that
  \begin{enumerate}
  \item for $0\leq j\leq J$, the disk $D_j$ is regular and one has
    $\im p(z)<0$ in $D_j$;
  \item $z_1\in D_0$, $z_2\in D_J$, and $\psi$ has the standard
    behavior~\eqref{st_beh} in $D_0$;
  \item for $0\leq j\leq J$, there exists two solutions $\psi^j_\pm$ having
    the standard behavior
    $\psi_\pm^j\sim\frac1{\sqrt{\sin p(z)}} e^{\pm \frac
      ih\int_{z_0}^z p d\zeta}$ in the domain $D_j$.
  \end{enumerate}
  Denote the rightmost point of the boundary of $D_j$ by $w_j$.
  Possibly, excluding some of the disks $D_j$ from the collection
  $(D_j)_{0\leq j\leq J}$ and reordering them, we can and do assume that, for
  $1\leq j\leq J$, $D_{j}\setminus D_{j-1}\not=\emptyset$ and
  $w_{j}>w_{j-1}$. Indeed, to choose $D_1$, consider the point
  $w_0$. If $w_0$ is to the right of $z_2$, we can keep only $D_0$ in
  the collection. Otherwise, in our collection, there is a disk that
  contains $w_0$. Denote it by $D_1$. We then obtain the set of
  disks by induction.\\
  For $r>0$ we define
  \begin{equation*}
    S(r)=\{z\in \C\,:\,  |\im (z-z_1)|< r\}. 
  \end{equation*}
  For $1\le j\le J$, let $r_j=\min\{r>0\,:\, D_j\cap D_{j-1}\subset S(r)\}$.
  Pick $\D 0<\delta <\min_{1\le j\le J} r_j$ sufficiently small so
  that $I_\delta$, the $\delta$-neighborhood of $I$, be a subset of
  $\cup_{j=0}^J D_j$.\\
  Let us prove that $\psi$ has the standard behavior~\eqref{st_beh} in
  $I_\delta$. \\
  First we note that, for $h$ sufficiently small, by means of the
  formula $\psi(z)=-\psi(z-2h)-v(z)\psi(z-h)$ (i.e., by means of
  \cref{main}), $\psi$ can be analytically continued in $I_\delta$.
  It clearly satisfies equation~\eqref{main} in $I_\delta$.\\
  Let us justify the asymptotic
  representation~\eqref{standard_asymptotic} in $I_\delta$. For
  $0\le j\le J$, we define $d_j=D_j\cap I_\delta$.  For $j=1,2,\dots J$,
  we consecutively prove that $\psi$ has the standard behavior
  in $d_j$ to the right of $d_{j-1}$. Therefore, we let
  $d^0=d_{j-1}$, \ $d=d_{j}$, \ $\psi_\pm=\psi_\pm^{j}$, and then 
  proceed in the following way.\\
  Using the standard asymptotic behavior of $\psi_+$ and $\psi_-$, one
  proves the asymptotic formula
  \begin{equation}
    \label{w:f-pm}
    w(\psi_+(z),\psi_-(z))=2i+o(1), \quad z, z+h\in D_{j},\qquad h\to 0.
  \end{equation}
  As the Wronskians are $h$-periodic, for sufficiently small $h$,
  formula~\eqref{w:f-pm} is valid uniformly in $ I_{\delta}$.  This
  implies in particular that, for sufficiently small $h$, in
  $I_\delta$, the solution $\psi$ is a linear combination of the
  solutions $\psi_\pm $ with $h$-periodic coefficients, and one
  has~\eqref{eq:three-solutions} and~\eqref{eq:periodic-coef}.\\
  The leading terms of the asymptotics of $\psi$ and $\psi_+$ coincide
  in $d^0\cap d$.  Thus, one has
  \begin{equation}
    \label{6.1}
    a(z)=1+o(1), \quad z, z+h\in d^0\cap d,\qquad h\to 0.
  \end{equation}
  Due to the $h$-periodicity of $a$, for sufficiently small $h$,
  formula~\eqref{6.1} stays valid  in the whole $I_\delta$.\\
  One also has
  $ b(z)= o\left(e^{\frac{2i}h\int_{z_0}^{z}p(\zeta) d\zeta}\right)$
  for $z\in d^0\cap d$.  For sufficiently small $h$, the $h$-periodicity
  of $b$ yields
  \begin{equation}
    \label{eq:1}
    b(z)=o\left(e^{\frac{2i}h\int_{z_0}^{\tilde z}p(\zeta)
        d\zeta}\right),\quad z\in I_\delta,
  \end{equation}
  where $\tilde z\in d^0\cap d$ and $\tilde z=z$ mod $h$.\\
  Estimates~\eqref{6.1} and~\eqref{eq:1} imply that, in $d$, one has
  \begin{equation*}
    \psi(z)=a(z)\psi_+(z)+b(z)\psi_-(z)=\frac{e^{\frac
        ih\int_{z_0}^{z} p(\zeta)d\zeta}}{\sqrt{\sin p(z)}} \left (1+o(1)+
      o\left(e^{-\frac{2i}h\int_{\tilde z}^{z}p(\zeta)
          d\zeta}\right)\right).
  \end{equation*}
  Assume that $z\in d$ is located to the right of $d_0$.  As $\im p<0$ in
  $d$, one has
  $\re\left(i\int_{\tilde z}^{z} p(\zeta)d\zeta\right)>0$. This
  implies that $\psi$ has the standard behavior~\eqref{st_beh} in $d$
  to the right of $d^0$ and completes the proof of
  Theorem~\ref{th:rectangle}.
\end{proof}
\section{The asymptotics outside a neighborhood of 0}
\label{s:outside-as}
We consider the solution $\psi$ described in~\cref{sss:main} and
derive its asymptotics outside a neighborhood of $0$, the pole of $v$.
\subsection{The asymptotics outside a neighborhood of $\R_+$}
\label{sec:asympt-outs-neghb}
Recall that
\begin{itemize}
\item in the rectangle $S_c$, the solution $\psi$ has the standard
  behavior~\eqref{st_beh};
\item $S'$ is regular;
\item in $S'$, the branch $p$ appearing in~\eqref{st_beh} satisfies
  the inequality $\im p(z)<0$.
\end{itemize}
\Cref{th:rectangle} yields the asymptotics of $\psi$ in $S'$
to the right of $S_c$, namely
\begin{lemma}
  \label{le:psi-in-S-prime}
  The solution $\psi$ admits the standard asymptotic
  behavior~\eqref{st_beh} in the domain $S'$.
\end{lemma}
\noindent Let us underline that the obstacle to justify the standard
behavior of $\psi$ in the whole domain $S$ is the pole of $v$ at 0.
\subsection{Asymptotics in a neighborhood of $\R_+$ outside a
  neighborhood of 0}
\label{sec:asympt-neghb-r_+}
We note that the function $z\mapsto (1-e^{2\pi i z/h})$ is
$h$-periodic. As $\psi$ satisfies~\cref{main}, so does
$z\mapsto (1-e^{2\pi i z/h})\,\psi(z)$. Moreover, as $\psi$ has poles
only at the points of $h\N$ and as these poles are simple, the
solution $z\mapsto (1-e^{2\pi i z/h})\,\psi(z)$ is analytic in $S$.\\
For $\delta>0$, let $P(\delta)=\{z\in S\,:\, \re z>0, \, |\im z|< \delta\}$.  In this
subsection, we prove
\begin{proposition}
  \label{pro:in-out}
  Let $\delta>0$ be sufficiently small. In $P(\delta)$, the solution
  $z\mapsto (1-e^{2\pi i z/h})\,\psi(z)$ has the standard behavior
  \begin{equation}
    \label{eq:psi-factor}
    (1-e^{2\pi i z/h})\,\psi(z)\sim n_0\,
    \frac{e^{\textstyle\frac ih\int_0^z p_{up}(\zeta)\,d\zeta}}{\sqrt{\sin p_{up}(z)}}\,, 
    \qquad n_0=e^{\frac{i}{h}\int_{z_0}^0 p(z)\,dz};
  \end{equation}
  here, $p_{up}$ and $\sqrt{\sin p_{up}}$ are respectively obtained
  from $p$ and $\sqrt{\sin p}$ by analytic continuation from
  $S'\cap\C_+$ to $P(\delta)$.
\end{proposition}
\noindent To prove Proposition~\ref{pro:in-out}, it suffices to check
that, for any point $z_*\in P(\delta)$, there exists a neighborhood, say,
$V_*$ of this point (independent of $h$) where the solution
$z\mapsto (1-e^{2\pi i z/h})\psi(z)$ has the standard behavior.\\
As their study is simpler, we begin with the points $z_*\not\in \R$.
\subsubsection{Points $z_*$ in $\C_+$}
\label{sec:case-where-im}
Let $z_*\in P(\delta)\cap \C_+$. Let $V_*\subset P(\delta)\cap\C_+$ be an open disk (independent
of $h$) centered at $z_*$. In $V_*$
one has $1-e^{2\pi i z/h}=1+o(1)$ as $h\to0$.  Furthermore, by
\Cref{le:psi-in-S-prime}, $\psi$ has the standard
behavior~\eqref{st_beh} in $V_*$. Therefore, in $V_*$, one computes
\begin{equation}
  \label{eq:psi-with-factor}
  (1-e^{2\pi i z/h})\,\psi(z)= \frac{e^{\frac{i}{h}\int_{z_0}^z
      p(\zeta)\,d\zeta+o(1)}}{\sqrt{\sin p(z) }}=
  n_0\,\frac{e^{\frac{i}{h}\int_{0}^z
      p_{up}(\zeta)\,d\zeta+o(1)}}{\sqrt{\sin p_{up}(z)}}. 
\end{equation}
This implies the standard behavior~\eqref{eq:psi-factor} in $V_*$.
\subsubsection{Points $z_*$ in $\C_-$}
\label{sec:case-where-im-1}
Pick now $z_*\in P(\delta)\cap \C_-$ and let $V_*\subset P(\delta)\cap\C_-$ be an open disk
(independent of $h$) centered at $z_*$.\\
We use
\begin{lemma}
  \label{le:p_up-down}
  For $z\in P(\delta)\cap \C_-$, one has
  \begin{equation}
    \label{eq:p-p-d}
    p(z)=p_{up}(z)-2\pi, \quad \sqrt{\sin p(z)}=-\sqrt{\sin p_{up}(z)}.
  \end{equation}
\end{lemma}
\noindent Lemma~\ref{le:p_up-down} yields
\begin{equation}
  \label{eq:123}
  \frac{e^{\frac{i}{h}\int_{z_0}^z p(\zeta)\,d\zeta}}{\sqrt{\sin p(z)}} =
  -n_0 \,\frac{e^{\frac{i}{h}\int_{0}^z
      (p_{up}(z)-2\pi)\,dz}}{\sqrt{\sin p_{up}(z)}},\quad z\in V_*.
\end{equation}
By~\Cref{le:psi-in-S-prime}, $\psi$ has the standard
behavior~\eqref{st_beh} in $V_*$. Therefore, \eqref{eq:123} implies
that, in $V_*$, one has
\begin{equation*}
  (1-e^{2\pi i z/h})\psi(z) = n_0\,(1-e^{-2\pi i z/h})
  \frac{e^{\frac{i}{h}\int_{0}^zp_{up}(\zeta)\,d\zeta+o(1)}}{\sqrt{\sin p_{up}(z)}}=
 n_0\,\frac{e^{\frac{i}{h}\int_{0}^zp_{up}(\zeta)\,d\zeta+o(1)}}{\sqrt{\sin p_{up}(z)}},
\end{equation*}
and $z\mapsto (1-e^{2\pi i z/h})\psi(z)$ has the standard 
behavior~\eqref{eq:psi-factor}  in $V_*$.\\
To prove \Cref{le:p_up-down}, we shall use
\begin{lemma}\label{le:p-near-zero}
  In $S'$, one has
  \begin{equation}
    \label{eq:p-as-0}
    p(z)=i\ln(z)+C+g(z)
  \end{equation}
  where $\ln$ is a branch of the logarithm analytic in
  $\C\setminus\R_+$, $C$ is a constant, and $g$ is a function analytic
  in $S\cup{0}$ vanishing at $0$.
\end{lemma}
\begin{proof}[Proof of \Cref{le:p-near-zero}]
  By definition (see~\eqref{eq:p}), $p$ satisfies
  $ e^{2ip(z)}+v(z)e^{ip(z)}+1=0$. Therefore,
  \begin{equation*}
    e^{ip(z)}=-v(z)/2+\sqrt{ (v(z)/2)^2-1},
  \end{equation*}
  where the branch of the square root is to be determined.  Since
  $v(z)\to\infty$ as $z\to0$, we rewrite this formula in the form
  \begin{equation}
    \label{exp-i-p}
    e^{ip(z)}=-v(z)/2\,(1+\sqrt{1-(2/{v(z)})^2}).
  \end{equation}
  As $\im p(z)<0$ in $S'$ and as $v(z)\to\infty$ when $z\to0$,
  \cref{eq:p} implies that $e^{ip(z)}\to\infty$ when $z\to 0$. Therefore,
  in~\eqref{exp-i-p}, the determination of the square root is to be
  chosen so that $\sqrt{1-(2/{v(z)})^2}=1+o(1/(v(z))^2)$ as
  $z\to 0$. Then,~\eqref{exp-i-p} yields the representation
  \begin{equation}\label{exp-i-p:1}
    e^{ip(z)}=-v(z)+\tilde g(z)
  \end{equation}
  where $\tilde g$ is analytic in a neighborhood of 0 vanishing at
  $0$. By assumption, $v$ has the Laurent expansion
  $v(z)=v_{-1}/z +v_0+v_1 z+\dots$, \ $v_{-1}\ne 0$, in a neighborhood
  of 0. Thus, \eqref{exp-i-p:1} implies
  representation~\eqref{eq:p-as-0} in a neighborhood of 0.\\
  The function $z\mapsto g(z):=p(z)-i\ln(z)-C$ is analytic in $S$ in a
  neighborhood of $\R_-$.  To check that it is analytic in the whole
  domain $S$, we consider two analytic continuations of
  $z\mapsto p(z)-i\ln z -C$ into a neighborhood of $\R_+$ in $S$, one
  from $S'\cap\C_+$ and another from $S'\cap \C_-$. As they coincide
  near 0, they coincide in the whole connected component of $0$ in
  their domain of analyticity. So, $g$ is analytic in $S$.  This
  completes the proof of \Cref{le:p-near-zero}.
\end{proof}
\noindent It now remains to prove Lemma~\ref{le:p_up-down}. Therefore, 
we first prove  \Cref{le:p-ln}.
\begin{proof}[The proof of \Cref{le:p-ln}]
  The statements of \Cref{le:p-ln} on the analyticity of the functions
  $z\mapsto p(z)-i\ln(-z)$ and $z\mapsto z\sin p(z)$ follow directly
  from \Cref{le:p-near-zero}. This lemma also implies that the second
  function does not vanish at 0. Finally, this function does not
  vanish in $S\setminus \{0\}$ in view of \Cref{rem:sin}.  The proof
  of \Cref{le:p-ln} is complete.
\end{proof}
\begin{proof}[Proof of \Cref{le:p_up-down}]
  Let $z\in P(\delta)\cap \C_-$. The first formula in~\eqref{eq:p-p-d}
  follows directly from the representation~\eqref{eq:p-as-0}. By
  \Cref{le:p-ln}, the function $z\mapsto \sin p(z)$ is analytic and
  does not vanish in $S\setminus\{0\}$. So, for
  $z\in P(\delta)\cap \C_-$, \ $\sqrt{\sin p(z)}$ and
  $\sqrt{\sin p_{up}(z)}$ either coincide or are of opposite signs.
  In view of~\eqref{eq:p-as-0}, we obtain the second relation
  in~\eqref{eq:p-p-d}.
\end{proof}
\noindent Having proved \Cref{le:p_up-down}, the analysis for $z\in \C_-$ is
complete.
\subsubsection{Real points $z_*$: construction of two linearly
  independent solutions}
\label{sec:case-z_inr.-two}
To treat the case of real points, we define two linearly independent
solutions to~\cref{main} that have standard asymptotic behavior to the
right of 0 and express $\psi$ in terms of these solutions.\\
Below, in the proof of \Cref{pro:in-out}, we always assume that $z_*$
is a point in $\R_+^*\cap S$. Let $x_\pm\in S$ be two points such
that $0<x_+<z_*<x_-$.\\
We recall that the set $S\setminus\{0\}$ is regular (see the beginning
of \cref{ss:main}). For $\bullet\in\{+,-\}$, the point $x_\bullet$ is
regular.  By \Cref{le:1} and \Cref{th:wkb_main}, there exists a
regular $\epsilon_\bullet$-neighborhood $V_\bullet$ of $x_\bullet$ in
$P(\delta)$ such that there exists a solution $\psi_\bullet$
to~\eqref{main} with the standard asymptotic behavior
$\psi_\bullet\sim \frac1{\sqrt{\sin p_{up}(z)}}e^{\bullet\frac {i}h
  \int_{x_\bullet}^z p_{up}(\zeta)\,d\zeta}$ in the domain
$V_\bullet$.\\
Let $\epsilon=\min\{\epsilon_+, \epsilon_-\}$ and $V$ be the
$\epsilon$-neighborhood of the interval $(x_+,x_-)$.\\
As the imaginary part of the complex momentum can not vanish in
$S\setminus \{0\}$, $\im p_{up}$ is negative in $V$. Therefore, by
\Cref{th:rectangle}, the solution $\psi_+$ has the standard asymptotic
behavior
\begin{equation}
  \label{eq:3}
 \psi_+(z)\sim \frac{e^{\frac{i}{h}\int_{x_+}^z
    p_{up}(\zeta)\,d\zeta}}{\sqrt{\sin(p_{up}(z))}} 
\end{equation}
in $V$ to the right of $V_+$.\\
Similarly, one shows that $\psi_-$ has the standard asymptotic
behavior
\begin{equation}
  \label{eq:4}
 \psi_-(z)\sim \frac{e^{-\frac{i}{h}\int_{x_-}^z
    p_{up}(\zeta)\,d\zeta}}{\sqrt{\sin(p_{up}(z))}} 
\end{equation}
in $V$ to the left of $V_-$.\\
Then,~\eqref{eq:3} and~\eqref{eq:4} yield that, as $h\to 0$, one has
\begin{equation}
  \label{eq:w_psi_pm}
  w(\psi_+(z),\psi_-(z))=2ie^{\frac{i}{h}\int_{x_+}^{x_-}p_{up}(z)\,dz+o(1)}, \qquad z\in V.
\end{equation}
At cost of reducing $\varepsilon$, this asymptotic is uniform in
$V$. We note that the error term is analytic in $z$ together with
$\psi_\pm$.\\
In view of~\eqref{eq:w_psi_pm}, for sufficiently small $h$, the
solutions $\psi_\pm$ form a basis of the space of solutions
to~\cref{main} defined in $V$; for $z\in V$, we
have~\eqref{eq:three-solutions} and~\eqref{eq:periodic-coef}.  Our
next step is to compute the asymptotics of $a$ and $b$
in~\eqref{eq:three-solutions}.
\subsubsection{The coefficient $a$}
\label{sec:coefficient-a}
We prove
\begin{lemma}
  \label{le:2}
  As $h\to 0$,
  \begin{equation}
    \label{as:a}
    a(z)=\frac{n_0\;e^{\frac{i}{h}\int_{0}^{x_+} p_{up}(z)\,dz+o(1)}}
    {1-e^{2\pi i z / h}}, \quad z\in V,
  \end{equation}
  where the error term is analytic in $z$.
\end{lemma}
\begin{proof}
  For $z\in V$, we shall compute the asymptotics of the Wronskian
  $w_-(z)=w(\psi(z),\psi_-(z))$ appearing in the formula for $a$
  in~\eqref{eq:periodic-coef}.\\
  As those of $\psi$, the poles of $w_-$ in $S$ are contained in $h\N$
  and they are simple. We first compute the asymptotics of $w_-$ in
  $V$ outside the real line.  Then, the information on the poles
  yields a global asymptotic representation for $w_-$ in
  $V$ and, thus,~\eqref{as:a}.\\
  First, we assume that $z\in V\cap \C_+$. Then, $p_{up}$ and
  $\sqrt{\sin p_{up}} $ coincide respectively with $p$ and
  $\sqrt{\sin p}$; using the asymptotics of $\psi$ and $\psi_-$ yields
  \begin{equation}
    \label{a:w-up}
    w_-(z)=2in_0\; e^{\frac{i}{h}
      \int_{0}^{x_-} p_{up}(z)\,dz+o(1)}, 
    \quad z\in V\cap \C_+, \quad h\to 0.
  \end{equation}
  Now, we assume that $z\in V\cap \C_-$. Then,~\eqref{eq:p-p-d} implies
  that
  \begin{equation}
    \label{eq:5}
    \psi(z)=-\frac{n_0\;e^{\frac
        ih\int_{0}^z(p_{up}(\zeta)-2\pi)\,d\zeta
        +o(1)}} 
    {\sqrt{\sin p_{up}(z)}}, \qquad h\to0.
  \end{equation}
  This representation and~\eqref{eq:4} yield
  \begin{equation}
    \label{a:w-down}
    w_-(z)=-2 i n_0\;e^{-2\pi i z/h}
    e^{\frac ih\int_{0}^{x_-}p_{up}(z)\,dz}(1+o(1)), 
    \quad z\in V\cap \C_-, \quad h\to 0.
  \end{equation}
  Let
  \begin{equation*}
    f\,:\,z\mapsto n_0^{-1}\,(1-e^{2\pi iz/h}) \, e^{-\frac ih
      \int_{0}^{x_-}p_{up}(z)\,dz}w_-(z)-2i.  
  \end{equation*}
  Representations~\eqref{a:w-up} and~\eqref{a:w-down} imply that
  \begin{equation}
    \label{f:est}
    f(z)=o(1),\quad z\in V\setminus \R,\quad h \to 0.
  \end{equation}
  We recall that $f$ is an $h$-periodic function (as is $w_-$).
  Therefore, at the cost of reducing $\epsilon$ somewhat, we get the
  uniform estimate $f(z)=o(1)$ for $|\im z|=\epsilon$ and $h$
  sufficiently small. Moreover, the description of the poles of $w_-$
  implies that $f$ is analytic in the strip $\{|\im z|\le \epsilon\}$.
  \\
  Now, let us consider $f$ as a function of $\zeta=e^{2\pi i z/h}$.
  It is analytic in the annulus
  $\{e^{-2\pi\epsilon/h}\le |\zeta|\le e^{2\pi\epsilon/h}\}$ and, on
  the boundary of this annulus, it admits the uniform estimate
  $|f(\zeta)|=o(1)$. By the maximum principle, the estimate
  $f(\zeta)=o(1)$ holds uniformly in the annulus.\\
  Thus, as a function of $z$, the function $f$ satisfies the uniform
  estimate $ f(z)=o(1)$ for $|\im z|\le\epsilon$ and, therefore, in
  the whole domain $V$.
  \\
  This estimate, the representation~\eqref{eq:w_psi_pm} and the
  definition of $a$ (see~\eqref{eq:periodic-coef}) imply~\eqref{as:a}.
\end{proof}
\subsubsection{The coefficient $b$}
\label{sec:coefficient-b}
We estimate $b$ in $V$ for $h$ sufficiently small. To state our
result, let $\gamma$ be the connected component of $x_+$ in the set of
$z\in P(\delta)$ satisfying
\begin{equation*}
  \im \int_{x_+}^{z}p(z)\,dz=0.
\end{equation*}
As
\begin{equation*}
  \frac{\partial}{\partial x}\im \int_{x_+}^{z}p(\zeta)\,d\zeta
  =\im p(z)\ne 0\quad \text{in } P(\delta),
\end{equation*}
the Implicit Function Theorem guarantees that $\gamma$ is a smooth
vertical curve in a neighborhood of $x_+$. Reducing $\epsilon$ if
necessary, we can and do assume that $\gamma$ intersects both the
lines $\{\im z=\pm \epsilon\}$. We prove
\begin{lemma}
  \label{le:3}
  In $V$ (with $\epsilon$ reduced somewhat if necessary), on $\gamma$
  and to the right of $\gamma$, one has
  \begin{equation}
    \label{est:b}
    b(z)=o(1)\; \frac{n_0\;e^{\frac{i}h\int_{0}^{x_+} p_{up}(z)\,dz-
        \frac{i}h\int_{x_+}^{x_-} p_{up}(z)\,dz}}{1-e^{2\pi i z / h}},
    \quad h\to 0, 
  \end{equation}
  where $o(1)$ is analytic in $z$.
\end{lemma}
\begin{proof}
  Let us estimate the Wronskian $w_+(z)=w(\psi_+(z),\psi(z))$, the
  numerator  in~\eqref{eq:periodic-coef}.\\
  First, we assume that $z\in V\cap\C_+$.  Then, $p_{up}$ and
  $\sqrt{\sin p_{up}} $ coincide respectively with $p$ and
  $\sqrt{\sin p}$.  Thus, the leading terms of the asymptotics of
  $\psi$ and $\psi_+$ coincide up to a constant factor; this yields
  \begin{equation}
    \label{b:w-up:1}
    w_+=o(1)\,n_0\;e^{\frac{i}h\int_{0}^{x_+}p_{up}(z)\,dz+\frac{2i}h
      \int_{x_+}^{z}p_{up}(\zeta)\,d\zeta},\qquad  h\to 0.
  \end{equation}
  We recall that the Wronskians are $h$-periodic (see
  \cref{s:space-of-sol}). Let us assume additionally that $z$ is
  either between $\gamma$ and $\gamma+h$ or on one of these curves.
  Pick $\tilde z\in\gamma$ such that $\im \tilde z=\im z$.  In view of
  the definition of $\gamma$, as $p_{up}$ is analytic in $P(\delta)$,
  one has
  \begin{equation*}
    \left|e^{\frac{2i}h\int_{x_+}^{z}p_{up}(z)\,dz}\right|= 
    \left|e^{\frac{2i}h\int_{\tilde
          z}^{z}p_{up}(\zeta)\,d\zeta}\right|\le e^C,
    \quad h\to0;
  \end{equation*}
  here, $C$ is a positive constant independent of $h$.  This estimate
  and~\eqref{b:w-up:1} imply that, for $z\in V\cap \C_+$ either between
  the curves $\gamma$ and $\gamma+h$ or on one of them, one has
  \begin{equation}
    \label{b:w-up:2}
    w_+=o(1)\,n_0\;e^{\frac{i}h\int_{0}^{x_+}p_{up}(z)\,dz},\qquad h\to 0.
  \end{equation}
  Reducing $\epsilon$ somewhat if necessary, we can and do assume
  that~\eqref{b:w-up:2} holds on the line $\im z=\epsilon$ between
  $\gamma$ and $\gamma+h$ or on one of these curves.  Then, thanks to
  the $h$-periodicity of $w_+$, it holds for all $z$ on the line
  $\{\im z=\epsilon\}$.\\
  Now, we assume that $z\in V\cap\C_-$. Using the asymptotics~\eqref{eq:5}
  and~\eqref{eq:3}, we compute
  \begin{equation*}
    w_+(z)=o(1)\,n_0\;e^{-2\pi i z/h}\;
    e^{\frac{i}h\int_{0}^{x_+}p_{up}(z)\,dz+
      \frac{2i}h\int_{x_+}^{z}p_{up}(\zeta)\,d\zeta}
    \qquad h\to 0.
  \end{equation*}
  Arguing as when proving~\eqref{b:w-up:2}, we finally obtain
  \begin{equation}\label{b:w-down}
    w_+(z)=o(1)\,n_0\;e^{-2\pi i
      z/h}\;e^{\frac{i}h\int_{0}^{x_+}p_{up}(z)\,dz},  
    \quad \im z=-\epsilon,\quad h\to 0.
  \end{equation}
  Let $g\,:\,z\mapsto (1-e^{2\pi iz/h}) w_+(z)$.
  Estimates~\eqref{b:w-up:2} and~\eqref{b:w-down} imply that, for
  $|\im z|=\epsilon$, 
  \begin{equation}
    \label{w-plus:est}
    n_0^{-1}\,e^{-\frac{i}h\int_{0}^{x_+}p_{up}(z)\,dz}g(z)=o(1),\quad
    h\to 0.
  \end{equation}
  As those of the solution $\psi$ do, the poles of $w_+$ in
  $P(\delta)$ belong to the set $h\N$ and are simple. So, the function
  $g$ is analytic in the strip $\{|\im z|\leq\epsilon\}$. Moreover, it
  is $h$-periodic as $w_+$ is.  Thus, the maximum principle implies
  that~\eqref{w-plus:est} holds in the whole strip
  $\{|\im z|\le \epsilon\}$. Estimate~\eqref{w-plus:est},
  representation~\eqref{eq:w_psi_pm} and the definition of $b$
  (see~\eqref{eq:periodic-coef}) yield~\eqref{est:b}. This completes
  the proof of Lemma~\ref{le:3}.
\end{proof}
\subsubsection{Completing the proof of~\Cref{pro:in-out}}
\label{sec:compl-proof-crefpr}
Let $V_*\subset V$ be a disk independent of $h$, centered at $z_*$ and
located to the right of $\gamma$ (i.e., such that, for any $z\in V_*$,
there exists $\tilde z\in \gamma$ such that $\re \tilde z<\re z$ and
$\im \tilde z=\im z$).
\\
Below we assume that $z\in V_*$.
\\
Using~\eqref{eq:4} and~\eqref{eq:3}, the asymptotic representations
for $\psi_\pm$, and~\eqref{as:a} and~\eqref{est:b}, the
representations for $a$ and $b$, we compute
\begin{equation*}
  \frac{b(z)\,\psi_-(z)}{a(z)\psi_+(z)}=o(1)
  e^{-\frac{2i}h\int_{x_+}^{z}p_{up}(\zeta)d\zeta},\quad h\to 0.
\end{equation*}
As before, let $\tilde z\in\gamma$ be such that $\im\tilde z=\im z$.
Then, we have
\begin{equation*}
  \left|e^{-\frac{2i}h\int_{x_+}^{z}p_{up}(\zeta)d\zeta}\right|=
  \left|e^{-\frac{2i}h\int_{\tilde
        z}^{z}p_{up}(\zeta)d\zeta}\right|\le 1.
\end{equation*}
Here, we used the definition of $\gamma$ and the fact that
$\im p_{up}<0$ in $V$.  As a result, we have
$\frac{b(z)\,\psi_-(z)}{a(z)\psi_+(z)}=o(1)$.
Formula~\eqref{eq:three-solutions} then yields
\begin{equation*}
  \psi(z)=a(z)\,\psi_+(z)\left(1+\frac{b(z)\,\psi_-(z)}
    {a(z)\psi_+(z)}\right)=a(z)\,\psi_+(z)(1+o(1)).
\end{equation*}
where $o(1)$ is analytic in $z\in V_*$. This and the asymptotic
representations for $\psi_+$ and $a$ yields~\eqref{eq:psi-factor} in
$V_*$.   This completes the proof of \Cref{pro:in-out}.
\section{Global asymptotics}
\label{s:uniform-as}
\subsection{The proof of \Cref{th:main}}
\label{sec:proof-crefth:main}
We follow the plan outlined in \cref{ss:plan}. Recall that $G_0$ is
defined in~\eqref{eq:G0}. We prove
\begin{proposition}
  \label{pro:on-the-circle}
  Let $\delta>0$ be sufficiently small.  In $S$, outside the
  $\delta$-neighborhood of 0, one has
  \begin{equation}
    \label{eq:on-the-circle}
    \psi(z)/\Gamma(1-z/h)= n_0\,G_0(z)\,(1+o(1)),\quad h\to 0.
  \end{equation}
\end{proposition}
\noindent Let us check that \Cref{th:main} follows from Proposition~\ref{pro:on-the-circle}.\\
We recall that the poles of $\psi$ belong to $h\Z$ and are
simple.  Furthermore, in view of \Cref{le:p-ln}, $G_0$ is analytic in $S$.
Clearly, $G_0$ has no zeros in $S$. These observations imply that the function
\begin{equation*}
  f\,:\,z\ \longrightarrow \ \frac{\psi(z)}{n_0\,G_0(z)\,\Gamma(1-z/h)}-1
\end{equation*}
is analytic in $S$. By \Cref{pro:on-the-circle}, it satisfies the
estimate $f(z)=o(1)$ in $S$ outside the $\delta$-neighborhood of
0. Therefore, by the maximum principle, it satisfies this estimate in
the whole of $S$. This implies the statement of \Cref{th:main}.  To
complete the proof of this theorem, it now suffices to check  \Cref{pro:on-the-circle}.
\subsection{The proof of \Cref{pro:on-the-circle}}
\label{sec:proof-crefpr-circle}
Fix $\varepsilon$ sufficiently small positive. The proof of
\Cref{pro:on-the-circle} consists of two parts: first, we
prove~\eqref{eq:on-the-circle} in the sector
$S_\pi=\{z\in S\,:\, |z|\ge \delta,\ |\arg z-\pi|\le \pi-\varepsilon\}$,
and, then, in the sector
$S_0=\{z\in S\,:\, |z|\ge \delta,\ |\arg z|\le \varepsilon\}$.
\subsubsection{The asymptotic in the sector $S_\pi$}
\label{sec:asympt-sect-s_pi}
If $z\in S_\pi$ and $h\to 0$, we can use formula~\eqref{eq:stirling} for
$\Gamma(1-z/h)$ and the standard asymptotic
representation~\eqref{standard_asymptotic} for $\psi$, see
\Cref{le:psi-in-S-prime}. This immediately
yields~\eqref{eq:on-the-circle} in $S_\pi$.
\subsubsection{The asymptotic in the sector $S_0$}
\label{sec:asympt-sect-s_0}
Let $\varepsilon$ be so small that $S_0$ be a subset of $P(\delta)$
(defined just above \Cref{pro:in-out}). For $z\in S_0$ and $h$ small, we
express $\Gamma(1-z/h)$ in terms of $\Gamma(z/h)$ by
formula~\eqref{eq:two-gammas}, then, we use
formula~\eqref{eq:stirling} for $\Gamma(z/h)$ and the standard
asymptotic representation for $(1-e^{2\pi i/z})\psi(z)$ (see
\Cref{pro:in-out}). This yields
\begin{equation}
  \label{eq:psi-sur-gamma}
  \psi(z)/\Gamma(1-z/h)=n_0\,\tilde G_0(z) \,(1+o(1)),\quad h\to 0,\\ 
\end{equation}
where
\begin{equation}
  \label{eq:tildeG0}
  \tilde G_0(z)=\frac{i\;\sqrt{h/2\pi}}{\sqrt{z\sin p_{up}(z)}} \,  
  e^{\;\textstyle \frac{z}{h}\ln\frac1h+
    \frac{i}h\int_0^z(p_{up}(\zeta)-i(\ln(\zeta)-i\pi))\,d\zeta},
\end{equation}
and the functions $z\mapsto \ln z$ and $z\mapsto \sqrt{z}$ are
analytic in $\C\setminus\R_-$ and positive respectively if $z>1$ and $z>0$.  We
note that by \Cref{le:p-ln}, $\tilde G_0$ is analytic in $S_0$.\\
Define the functions $z\mapsto \sqrt{-z}$ and $z\mapsto \ln (-z)$ as
in~\eqref{eq:G0}, i.e. so that they be analytic in $\C\setminus\R_+$
and positive if $z<0$ and $z<-1$ respectively.  Then, these functions
are related to the functions $z\mapsto \ln z$ and $z\mapsto \sqrt{z}$
from~\eqref{eq:tildeG0} by the formulas
\begin{equation*}
  \sqrt{-z}=-i\sqrt{z},\quad \ln(-z)=\ln z-i\pi,\quad z\in \C_+.
\end{equation*}
Furthermore, in $S_0\cap \C_+$ the functions $p$ and $p_{up}$
coincide. These two observations imply that one has $\tilde G_0=G_0$
in $S_0\cap C_+$.

As both $G_0$ and $\tilde G_0$ are analytic in $S_0$, they coincide in
the whole of $S_0$. This and~\eqref{eq:psi-sur-gamma} imply the
representation~\eqref{eq:on-the-circle} for $z\in S_0$.  This completes
the analysis in the sector $S_0$ and the proof of
\Cref{pro:on-the-circle}.
\section{A basis for the space of solutions}
\label{s:basis}
We finally discuss a basis of the space of solutions to~\cref{main}
that are meromorphic in $S$. First, we describe the two solutions
forming the basis and, second, we compute their Wronskian.
\subsection{First solution}
\label{sec:first-solution}
As the first solution, we take $f_+(z)=\psi(z)/n_0$.  It has the
standard behavior~\eqref{eq:f-plus} in $S'$ and simple poles at the
points of $h\N$.  We recall that the modulus of the exponential factor
in~\eqref{eq:f-plus} increases when $z$ moves to the right parallel to
$\R$.
\subsection{Second solution}
\label{sec:second-solution}
Let $z_1>0$ be a point in $S$. Mutatis mutandis, in the way we
constructed $\psi$, in $S$ (possibly reduced somewhat), we construct a
solution $\phi$ in $S\setminus \R_-$ that has the standard asymptotic
behavior~\eqref{eq:phi} and such that the quasi-momenta $p$ (and the
functions $\sqrt{\sin p}$) in~\eqref{eq:f-plus} and~\eqref{eq:phi}
coincide in $\C_+$.
\\
The modulus of the exponential from~(\ref{eq:phi}) increases when $z$
moves to the left parallel to $\R$.
\\
The solution $\phi$ has simple poles at the points of $-h\N$ and, in
$S$, it admits the asymptotic representation
\begin{gather}
  \label{as:phi}
  \phi(z)=n_1\Gamma(1+z/h)\,G_1(z)(1+o(1)),\quad 
  n_1=e^{\frac{i}h\int_{z_1}^0 p(z)\,dz},\quad h\to 0,\\
  \label{eq:F0}
  G_1(z)=\frac{\sqrt{h/2\pi}}{\sqrt{z\sin p(z)}} \, e^{\;\textstyle
    -\frac{z}{h}\ln\frac1h-
    \frac{i}h\int_0^z(p(\zeta)-i\ln(\zeta))\,d\zeta}.
\end{gather}
Here, the functions $z\mapsto \sqrt{z}$ and $z\mapsto\ln z$ are
analytic in $\C\setminus \R_-$ and positive, respectively, if $z>0$ and
$z>1$.  The factor $G_1$ is analytic in $S$.
\\
We define the second solution to be
$f_-(z)= 1/n_1\;(1-e^{2\pi i z/h})\,\phi (z)$. The solution $f_-$ is
analytic in $S$. It has simple zeros at the points $z\in h\N$ and at 0.
By means of~\eqref{as:phi}, one can easily check that, in $S'$, it has
the standard behavior~\eqref{eq:f-minus}.
\subsection{The Wronskian of the basis solutions}
\label{sec:wronsk-basis-solut}
Using the asymptotic representations for $f_{\pm}$ in $S'$, one easily
computes
\begin{equation}
  \label{as:wf}
  w(f_+(z), f_-(z))=2i+o(1),\quad h\to 0.
\end{equation}
As the Wronskian is $h$-periodic, this representation is valid in the
whole domain $S$.  We see that, for sufficiently small $h$, the
leading term of the Wronskian does not vanish, and, thus, $f_\pm$ form
a basis in the space of solutions to~\cref{main} in $S$ (possibly
reduced somewhat for~\eqref{as:wf} to be uniform).
\bibliographystyle{plain}
\end {document}